\documentclass[journal]{IEEEtran}

\usepackage{graphicx,cite}
\usepackage{amsmath, mathtools}
\usepackage{url}
\usepackage{subfigure}

\usepackage{amssymb}
\usepackage{array}
\usepackage{fixltx2e}
\usepackage{amsfonts}
\usepackage{lineno}
\usepackage{dsfont}
\usepackage{bbm}
\usepackage{multirow}
\usepackage{booktabs}
\usepackage{bigdelim}

\usepackage{soul}

\usepackage[dvips]{color}
\usepackage{epsf}
\usepackage{times}
\usepackage{epsfig}
\usepackage{amsxtra}
\usepackage{algorithmic}
\usepackage{algorithm}

\newtheorem{proposition}{Proposition}

\newtheorem{definition}{\bf Definition}

\begin{document}

%
\title{\huge Backhaul-Aware Interference Management in the Uplink of Wireless Small Cell Networks
}
\author{
\IEEEauthorblockN{Sumudu Samarakoon\IEEEauthorrefmark{1}, Mehdi Bennis\IEEEauthorrefmark{1}, Walid Saad\IEEEauthorrefmark{2}, and Matti Latva-aho\IEEEauthorrefmark{1} \\}
\IEEEauthorblockA{\small\IEEEauthorrefmark{1}Centre for Wireless Communications, University of Oulu, Finland, \\ email: \url{{sumudu,bennis,matti.latva-aho}@ee.oulu.fi} \\
\IEEEauthorrefmark{2}Electrical and Computer Engineering Department, University of Miami, Coral Gables, FL, USA, email: \url{walid@miami.edu}}
\vspace{-.5cm}
\thanks{The authors would like to thank the Finnish funding agency for technology and innovation, Elektrobit and Siemens Networks for supporting this work. This work has been performed in the framework of the LOCON Project, TEKES, Finland, the SHARING project under the Finland grant 128010, and supported by the U.S. National Science Foundation (NSF) under Grant CNS-1253731.}
}
\maketitle
\nopagebreak[4]

\begin{abstract}
The design of distributed mechanisms for interference management is one of the key challenges in emerging wireless small cell networks whose backhaul is capacity limited and heterogeneous (wired, wireless and a mix thereof).
In this paper, a novel, backhaul-aware approach to interference management in wireless small cell networks is proposed.
The proposed approach enables macrocell user equipments (MUEs) to optimize their uplink performance, by exploiting the presence of neighboring small cell base stations.
The problem is formulated as a noncooperative game among the MUEs that seek to optimize their delay-rate tradeoff, given the conditions of \emph{both} the radio access network and the -- possibly heterogeneous -- backhaul.
To solve this game, a novel, distributed learning algorithm is proposed using which the MUEs autonomously choose their optimal uplink transmission strategies, given a limited amount of available information. 
The convergence of the proposed algorithm is shown and its properties are studied.
Simulation results show that, under various types of backhauls, the proposed approach yields significant performance gains, in terms of both average throughput and delay for the MUEs, when compared to existing benchmark algorithms.
\end{abstract}

\begin{keywords}
	heterogeneous networks;
	capacity-limited backhaul;
	wired and wireless backhaul;
	reinforcement learning;
	game theory.
\end{keywords}

\section{Introduction}\label{sec:intro}

The demand for high-speed wireless access is increasing at an unprecedented pace, requiring  operators to explore new approaches in order to boost the network capacity and improve  the coverage of next-generation wireless cellular networks. 
Recently, deploying low-cost, low-power small cells (such as femtocells and picocells) and relay nodes underlaid over existing macro-cellular networks has emerged as a promising solution to deal with this increasing wireless traffic demands~\cite{pap:vikram08}. 
The introduction of such small cell networks (also known as heterogeneous networks or HetNets) has attracted considerable attention from the research community~\cite{IG00,LP00,DL00,GTF02,pap:amitabha12,pap:osvaldo10,pap:ping11,pap:ritesh10,pap:ivana11}.
Existing works have dealt with various issues related to interference management, coverage improvement, traffic offload, mobility and load balancing, among others \cite{onln:mark12}.

Unlike traditional macro-cells, in which a dedicated, high-capacity backhaul exists, small cells are expected to connect to existing, capacity-limited Internet protocol (IP) backhauls such as digital subscriber line (DSL).
Consequently, when designing resource management solutions for small cell networks (SCNs), one must factor in a variety of constraints pertaining to the backhaul such as its limited capacity or its potential heterogeneity (i.e., the co-existence of a wired and wireless backhaul~\cite{IG00,pap:ritesh10B}).
Remarkably, while a body of literature on small cell networks exists such as \cite{DL00,GTF02,pap:lobinger10,pap:johnp10,pap:guvenc11}, beyond~\cite{pap:osvaldo10,pap:ping11,pap:ritesh10,pap:ivana11} little work seems to have studied how the backhaul constraints impacts the resource management processes at the radio access level.
Indeed, most works assume perfect and reliable backhaul conditions with infinite capacity~\cite{pap:vikram08,pap:ivana11}, \cite{pap:antonio11,onln:cisco10}  or near-perfect backhaul solutions relying on out-of-band and over-the-air (OTA) links \cite{pap:wang10,onln:netkrom12}. 
In \cite{LP00}, the authors provide a coverage and interference analysis while focusing on spectrum allocation and interference mitigation for an orthogonal frequency-division multiple access (OFDMA) macro and femtocell scenario.
The authors in \cite{DL00} propose a novel cooperative scheduling scheme for a two-tier network in order to mitigate both uplink and downlink interference with range expansion.
In \cite{pap:osvaldo10}, the authors propose novel interference management techniques in  a two-tier network where femtocells are connected to the macro base station (MBS) via unreliable access links.
The work in \cite{pap:ping11} studies the requirements of inter-cell interference coordination in order to improve the quality of service of SCNs,  under different backhaul types (wired or wireless).  
A dynamic interference management algorithm is proposed for heterogeneous networks in \cite{pap:ritesh10}, focusing on resource negotiation over third-party backhaul connections and on interference management for quality of service via over-the-air signaling. 
In \cite{pap:ivana11}, a backhaul scheduling scheme based on  picocells' traffic demands with shared wireless backhaul is proposed.

Clearly, deploying novel solutions for the backhaul of SCNs is critical in order to reap their potential capacity improvements.
However, these capacity improvements are essentially limited by the choice of an appropriate backhaul.
In fact, the heterogeneous and unreliable nature of the SCN's backhaul poses a fundamental question: should the small cells use an over-the-air (in-band) backhaul which requires significant spectrum resources but can guarantee reasonable delays, or should they use a third-party wired backhaul which does not require any spectrum resources but could lead to significant traffic delays?
The answer to this fundamental question is particularly important for an optimal small cell deployment that takes into account both access and backhaul links.

In this work, we study the problem of backhaul-aware resource allocation in the uplink of wireless SCNs.
Resource allocation in small cell networks requires a fundamental paradigm shift from traditional, centralized cellular system operation toward distributed, self-organizing solutions, collectively known as self-organizing networks (SONs).
The idea of SON has been recently studied in the context of  heterogeneous radio access networks (RANs), with a focus on self-optimization, self-configuration, and self-healing techniques \cite{book:hamalainen12,pap:viering09,pap:turkka11,pap:walid09}.
In \cite{pap:viering09}, the authors present a mathematical framework focusing on downlink load balancing via SON.
Those load balancing investigations are extended in \cite{pap:turkka11} for a non-regular cell layout in the uplink direction.
The authors in \cite{pap:walid09} propose a distributed formation of the uplink tree structure of relay nodes and BSs using a game theoretic approach. 
In \cite{GTF02}, the authors propose a hierarchical approach for power control in two-tier small cell networks, using multi-leader Stackelberg games.
Thus far, these studies have dealt with notions of self-optimization and self-configuration aiming at optimizing radio access parameters (power levels, frequency allocations, etc).

In order to address the problem of optimizing radio access parameters in a distributed manner, various mathematical tools have been proposed.
In particular, reinforcement learning (RL) has played a central role in designing self-organizing SCNs \cite{pap:carme12,pap:perlaza10,pap:bennis11,pap:bennis12}.
Using RL, wireless nodes are able to self-organize in a decentralized manner, relying only on  local information, with little reliance on centralized information exchange.
In \cite{pap:carme12}, the authors propose best-response dynamics in order to study the throughput and delay performance based on users and BS densities for both uplink and downlink. 
A RL mechanism is proposed in \cite{pap:perlaza10} for joint utility and strategy estimation for cognitive terminals and, extended in \cite{pap:bennis11} to provide an interference mitigation technique for cognitive femtocell networks.
The study conducted in \cite{pap:bennis12} introduces a regret-based learning algorithm for cognitive femtocell networks in order to mitigate cross-tier interference.
Consequently, given this recent interest in SON, in this paper, we aim at developing a novel, self-organizing interference management solution using which the macrocell users leverage existing small cells, while taking into account a  heterogeneous (wired, wireless) backhaul.

The main contribution of this paper is to propose novel self-organizing interference management strategies for wireless SCNs while taking into account  the constraints due to the presence of a heterogeneous backhaul. 
Here, small cells act as ``helping relays" by decoding and forwarding the macrocell user equipments (MUE) uplink traffic to the MBS over  heterogeneous 
backhauls\footnote{
Note that the downlink problem is a challenging problem which deserves its own investigation and thus, it is not discussed in this work.
}.
In the proposed approach, the MUEs judiciously split their uplink traffic into two parts; {\it (i) a coarse message} which can only be decoded at the MBS and {\it (ii) a fine message} which is broadcasted by the MUEs and decoded by neighboring small cell base stations~(SBS), as well as the MBS. 
Here, the MUEs  autonomously select their best helping SBS and optimize their transmission strategy (throughput/delay tradeoff), while at the same time accounting for the underlying backhaul conditions. 
Due to the coupling in MUEs' strategies, the problem is cast as a non-cooperative game with the MUEs as the players.
To find an equilibrium of the formulated game, we propose a solution using novel concepts from reinforcement learning, particularly, using learning with imperfect information.
Using the proposed approach, the MUEs self-organize and implicitly coordinate their transmission strategies in a fully distributed manner while optimizing their utility function which captures the tradeoff between rate and delay.
Simulation results show that using the proposed solution, the MUEs are able to optimize the tradeoff between throughput and delay, while significantly improving the overall network performance, relatively to a number of  benchmark interference management algorithms.

The rest of the paper is organized as follows. The network and backhaul models are introduced in Section~\ref{sec:sys_mod}. Section~\ref{sec:prop_mthd} presents the formulation of the proposed cooperative relaying technique over heterogeneous backhauls. The proposed reinforcement learning technique is discussed in Section~\ref{sec:game}. Simulation results are presented and analyzed in Section~\ref{sec:results}. Finally, conclusions are drawn in Section \ref{sec:conclu}.

\section{System Model} \label{sec:sys_mod}

\subsection{Notation}\label{sec:notation}
The regular symbols represent the scalars while the boldface symbols are used for the vectors.
The sets are denoted by upper case calligraphic symbols.
$|\mathcal{X}|$ and $\Delta(\mathcal{X})$ represent the cardinality and the set of all probability distributions of the finite set $\mathcal{X}$, respectively.
The function $\mathbbm{1}_{\mathcal{X}}(x)$ denotes the indicator function which is defined as,
$$\mathbbm{1}_{\mathcal{X}}(x)=\begin{cases} 1 &\mbox{if } x\in\mathcal{X}, \\ 0 &\mbox{if } x\notin\mathcal{X}.\end{cases}$$
Whenever a common parameter is formulated for different approaches, we use the notation of $[parameter]_{approach}$ in order to differentiate them from one other.
All the symbols which are used in the rest of the paper are summarized in Table \ref{tab:symbols}.

\begin{table}[!t]
\caption{Notation Summary.}
\begin{center}
\begin{tabular}{|p{23 mm} p{57 mm}|}
\hline
\makebox[23 mm][c]{\bf Symbols} & \makebox[57 mm][c]{\bf Description} \\
\hline \hline
$[\cdot]_{\rm CLA},~[\cdot]_{\rm WRD}$ and $[\cdot]_{\rm OTA}$	 & Property defined for classical approach, proposed approach with wired backhaul and over-the-air backhaul \\
$\mathcal{M},~\mathcal{K}$ and $\mathcal{S}$	& The Sets of MUEs, SUEs and SBSs \\
$\mathcal{N},~\mathcal{N}_{m},~\mathcal{N}_{k}$ and $\mathcal{N}_{s}$	& The set of sub-carriers and the sets of sub-carriers allocated for MUE $m$, SUE $k$ and SBS $s$ \\
$|h_{ji}^n|^2$	& Channel gain between transmitter $j$ and receiver $i$ over the sub-carrier $n$ \\
$P_j^n$	& Transmission power of the transmitter $j$ over sub-carrier $n$ \\
$I_k^n$	& The aggregated interference experienced by SUE $k$ \\
$N_0$	& Gaussian thermal noise power \\
$R_j,~D_j$	& Rate and Delay experienced by the UE (or the SBS act as a relay) $j$ \\
$\alpha$	& Sensitivity parameter between rate and delay \\
$C_s,~\overline{C}$	& Capacity share for SBS $s$ and the total capacity of the wired Backhaul \\
$\nu_{m,s}$	& The fraction of capacity allocation for UE $m$ from the backhaul of SBS $s$ \\
$\theta_m$	& The fraction of power allocation for fine message by the MUE $m$ \\
$R_{m,C},~R_{m,F}$	& Rate of the coarse message and the fine message for MUE $m$ \\
$D_{m,C},~D_{m,F}$	& Delay of the coarse message and the fine message for MUE $m$ \\
$\mathcal{A}_m$	& The set of actions of MUE $m$ \\
$a_m,~\boldsymbol{a}_{-m}$	& An action of MUE $m$ and the set of actions played by the remaining players \\
$u_m(a_m,\boldsymbol{a}_{-m})$	& The utility of MUE $m$ \\
$\boldsymbol{\pi}_m$	& The mixed-strategy probability distribution \\
$\boldsymbol{r}_m$	& The regret vector of MUE $m$ \\
$\tilde{u}_m,~\tilde{\boldsymbol{r}}_m$	& The feedback of the utility and the estimated regret vector of MUE $m$ \\
$\boldsymbol{\beta}_m(\cdot)$	& Boltzmann-Gibbs distribution \\
$\kappa_m$	& Temperature parameter which balance the exploration and exploitation in the learning algorithm \\
$\lambda,~\gamma$ and $\mu$	& Learning rates \\
\hline
\end{tabular}
\end{center}
\label{tab:symbols}
\end{table}

\subsection{Network Model}\label{sec:network_model}
Consider the uplink transmission of an OFDMA two-tier wireless small cell network.
The MBS is located at the center of a cell
and serves a set of MUEs denoted by $\mathcal{M} = \{1,\ldots,M\}$.
We let $\mathcal{S} = \{1,\ldots,S\}$ denote the set of SBSs that are underlaid on the macro-cellular network with each small cell $s\in \mathcal{S}$ having a radius of $d_s$.
Let $\mathcal{K}=\{1,\ldots,K\}$ denote the set of small cell users (SUEs) and let $\mathcal{N} = \{1,\ldots,N\}$ denote the set of sub-carriers.
Here, $|h_{ji}^{n}|^2$ represents the channel gain between transmitter $j$ and receiver $i$ on sub-carrier $n$.
We use the index $0$ to refer to the MBS.
The total transmit power of the $j$-th transmitter over $n$-th sub-carrier is $P_{j}^{n}$, and the variance of the complex Gaussian thermal noise at the receiver is denoted by $N_0$.

In the classical uplink transmission scenario, referred to as ``CLA'' hereafter,  no coordination is assumed between the macro- and small cell tiers. Here, the achievable rates of MUE $m\in\mathcal{M}$ and SUE $k\in\mathcal{K}$ serviced by the MBS and SBS $s$ are, respectively, given by:
\begin{align} \label{eqn:noncop_mbsrate}
	&\!\begin{multlined}[t][\displaywidth] 
		[R_{m}]_{\rm CLA} = \sum_{\forall n\in\mathcal{N}_m}\log_2 \biggl( 1 + \frac{}{N_0 + } \\
		\frac{|h_{m0}^{n}|^2P_m^n}{ \displaystyle\sum_{\substack{\forall i\in\mathcal{M}\\i{\neq}m}}\mathbbm{1}_{\mathcal{N}_i}(n)|h_{i0}^{n}|^2P_{i}^n + \displaystyle\sum_{\forall j\in\mathcal{K}}\mathbbm{1}_{\mathcal{N}_j}(n)|h_{j0}^{n}|^2P_{j}^n}  \biggr), 
	\end{multlined} \\
\label{eqn:noncop_SBSrate}
	&[R_{k}]_{\rm CLA} = \min\Biggl\{\sum_{\forall n\in\mathcal{N}_k}\log_2 \biggl( 1 + \frac{|h_{ks}^{n}|^2 P_{k}^n}{N_0 + I_k^n } \biggr), \nu_{s,k}C_s\Biggr\},
\end{align}
where $\mathcal{N}_m$ and $\mathcal{N}_k$ denote the set of sub-carriers assigned to MUE $m$ and SUE $k$, respectively.
The aggregate interference experienced by the $k$-th SUE is $I_k^n = \sum_{\forall i\in\mathcal{M}}\mathbbm{1}_{\mathcal{N}_i}(n)|h_{ms}^{n}|^2P_m^n + \sum_{\substack{\forall j\in\mathcal{K}\\j\neq k\\~}}\mathbbm{1}_{\mathcal{N}_j}(n)|h_{jf}^{n}|^2P_{j}^n$, the backhaul capacity between the $s$-th SBS and the MBS is~$C_s$, and $\nu_{s,k}$ is the fraction of the capacity allocation for SUE $k$. 
Moreover, without loss of generality, an M/D/1 queue with a packet generation rate of $\rho_m$, the transmission delay of the MUE $m$ is given by \cite{book:dim92}:
\begin{equation}\label{eqn:delayClasic}
	[D_m]_{\rm CLA} = \frac{\rho_m}{2[R_m]_{CLA}([R_m]_{CLA}-\rho_m)},
\end{equation}
where $[R_m]_{CLA}$ is given by (\ref{eqn:noncop_mbsrate}).

In this work, we propose a coordination mechanism among MUEs and SBSs in which the SBSs act as relay nodes for the MUEs.
Using this coordination, the users can improve their uplink transmission rates, as given in (\ref{eqn:noncop_mbsrate}).
However, this mandates not only an efficient coordination mechanisms but also an adequate backhaul design.
The main reason is the backhaul capacity is a limiting factor of the achievable throughput.
Although there exists a good over-the-air link between a user and a SBS, without a good backhaul link a user cannot achieve the expected higher rates.
In fact, the reliability of the backhaul connection between SBSs and MBSs is instrumental in the optimal deployment of  small cell networks.
Such networks requires designs that jointly account for access and backhaul links.
In practice, two types of backhauls are considered for wireless small cell networks~\cite{pap:wang10,onln:netkrom12}: wireless and wired, as described next.

\subsubsection{In-band wireless backhaul}
We distinguish two key types of wireless backhauls: {\it in-band backhaul} in which the backhaul network and all the users share the available entire spectrum band and {\it out-of-band backhaul} in which an additional separate spectrum band is allocated for the backhaul.
The availability of an out-of-band spectrum band is more convenient in providing a high capacity  backhaul.
However, it is not always possible for service providers to allocate an out-of-band spectrum due to limitations such as cost and spectrum availability~\cite{pap:carlos07}.
On the contrary, an in-band backhaul is always possible to integrate with the existing resources as long as it guarantees an improved service.
Thus, our focus is on the in-band wireless backhaul.

In this case, the over-the-air backhaul capacity of an arbitrary SBS $s$ -- the rate defined in (\ref{eqn:rateBackhaul}) -- is limited by the interference from other SBSs given by $\sum_{ \forall l\in\mathcal{S},l\neq s}\mathbbm{1}_{\mathcal{N}_l}(n)|h_{l0}^{n}(t)|^2P_{l}^n$.
Moreover, the rate of SBS $s$ over the backhaul link is given by:
\begin{equation}
	\label{eqn:rateBackhaul}
	R_{s0} = \sum_{\forall n\in\mathcal{N}_s} \log_2 \Biggl( 1 + \frac{|h_{s0}^{n}|^2P_{s}^n}{N_0 +  \displaystyle\sum_{ \forall l\in\mathcal{S},l\neq s}\mathbbm{1}_{\mathcal{N}_l}(n)|h_{l0}^{n}|^2P_{l}^n } \Biggr),
\end{equation}
where $\mathcal{N}_s$ is the set of sub-carriers allocated for the backhaul of the wireless small cell $s$.

\subsubsection{Wired backhaul}

Under a wired backhaul, we consider that the packet generation process of small cells follows a Poisson distribution in which we model the entire backhaul of the system as an M/D/$1$ queue~\cite{book:dim92}.
Let $C_s$ be the capacity of the $s$-th SBS-MBS link and, the total wired backhaul capacity $\overline{C}$ is given by:
\begin{equation}\label{eqn:totwiredcapacity}
	\sum_{\forall s\in\mathcal{S}}C_s \leq \overline{C},
\end{equation}
where the total capacity $\overline{C}$ is a fixed quantity and the capacity per SBS depends on the scenario.
It is worthy to mention that although the wired backhaul is interference-free unlike the over-the-air backhaul, it may suffer from congestion due to the fact that multiple small cells share the same backhaul.

\section{Cross-tier Coordination for SCNs}\label{sec:prop_mthd}

To enable an efficient co-existence between the macro and small cell tiers, we propose a cooperative approach using which the small cells assist neighboring MUEs to improve the overall performance, via the concept of rate splitting \cite{pap:randa11,pap:osvaldo10,pap:osvaldo09}.
In this context, each MUE $m \in  \mathcal{M}$ builds a {\it coarse} mesage $X_{m,C}^{n}$ and a {\it fine} message $X_{m,F}^{n}$ (direct signal and relayed signal, respectively) for each of its transmitted signals as illustrated in Fig.~\ref{fig:model}.
The source (MUE) superimposes those two codewords and broadcasts the combined message $X_m^n$.
The MBS is capable of decoding both coarse and fine messages.
Moreover, SBSs can only decode the MUE's fine message and relay it to the MBS.
Thus, the transmission rates associated with both coarse and fine messages are such that the neighboring SBSs can reliably decode the fine message, while the MBS decodes both the coarse and fine messages.
Mathematically, this is expressed as follows:
\begin{equation} \label{eqn:coarseNfine}
	 X_{m}^{n}=X_{m,C}^{n}+X_{m,F}^{n}.
\end{equation}
Moreover, the transmission power allocations of the MUE's coarse signal to the MBS and the fine signal to SBSs are $P_{m,C}^{n}=\big(1-\theta_m\big)P_m^n$ and $P_{m,F}^{n}=\theta_m P_m^n$, respectively, with $\theta_m\in[0,1]$ where $\theta_m$ is the fraction of power allocation for the fine message of MUE $m$.

\begin{figure}[!t]
\centering
\includegraphics[width=0.5\textwidth]{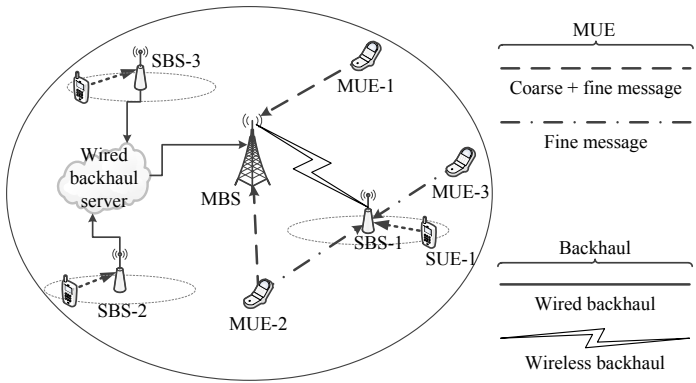}
\caption{Illustration of the proposed relaying approach in which the MUEs use rate splitting techniques in their uplink transmission. For an example MUE-2 uses SBS-1 as a relay while maintaining a direct link with the MBS. }
\label{fig:model}
\end{figure}

In what follows, we consider a capacity-limited backhaul which is shared among all SBSs.
Therefore, the amount of information that can be relayed and the number of MUEs which can be supported by each SBS are limited compared to a system with infinite backhaul capacity.
This is the main motivation behind the rate-splitting approach making use of both coarse and fine messages.

\subsection{Rate formulation with over-the-air (OTA) backhaul}

We assume a half-duplex decode-and-forward (DF) uplink transmission in which both the MUEs and the SUEs transmit during the first time slot and SBSs relay both signals (over the backhaul) during the second time slot. The uplink rate   of MUE $m\in \mathcal{M}$ when transmitting its coarse message to its serving MBS is given by:
\begin{equation}
\label{eqn:rateMacroPriv}
	[R_{m,C}]_{\rm OTA} = \sum_{n\in\mathcal{N}_m}\log_2 \Biggl( 1 + \frac{|h_{m0}^{n}|^2\big(1-\theta_m\big)P_m^n}{N_0 + |h_{m0}^{n}|^2\theta_m P_m^n + J_m^n } \Biggr),
\end{equation}
where $(1-\theta_m)P_m^n$ is the MUE's transmission power allocated for the coarse message and ``OTA" refers to the over-the-air backhaul model.
The interference caused by other users is given by $J_m^n = \sum_{\substack{\forall i\in\mathcal{M}\\i{\neq}m}}\mathbbm{1}_{\mathcal{N}_i}(n)|h_{i0}^{n}|^2P_{i}^n + \sum_{\forall j\in\mathcal{K}^n}|h_{j0}^{n}|^2P_{j}^n$.
The fine message is received by the MBS and relayed via an SBS\footnote{In this work, we assume that an MUE gets support from at most one SBS. The scenario of an MUE receiving support from multiple SBSs is left for future work.}. 
Thus, the rate of the MUE $m$-th fine message overheard by the MBS is calculated as,
\begin{equation}
\label{eqn:rateMacroCom}
	[R_{m0,F}]_{\rm OTA} = \sum_{n\in\mathcal{N}_m}\log_2 \Biggl( 1 + \frac{|h_{m0}^{n}|^2\theta_m P_m^n}{N_0 + J_m^n } \Biggr),
\end{equation}
Note that the sum of (\ref{eqn:rateMacroPriv}) and (\ref{eqn:rateMacroCom}) is equivalent to the MUE rate in the classical two-tier network scenario, as given in (\ref{eqn:noncop_mbsrate}).
Similarly, the rate of MUE $m$ when  transmitting its fine message to  SBS $s \in \mathcal{S}$ is given by:
\begin{equation}\label{eqn:ratesmall_cellCom}
	[R_{ms,F}]_{\rm OTA} = \sum_{n\in\mathcal{N}_m}\log_2 \Biggl( 1 + \frac{|h_{ms}^{n}|^2\theta_m P_m^n}{N_0 + I_k^n  } \Biggr),
\end{equation}
where the interference term $I_k^n$ is due to: a) the power used to transmit the coarse messages of MUE $m$ and b) the transmission power of other MUEs over sub-carrier $n$.

The SBS signal relayed over the wireless backhaul includes all the fine messages of the MUEs in which a rate fraction of $\nu_{s,m}(t) R_{s0}$ is allocated for the $m$-th MUE's fine message, where $\sum_{\forall m\in\mathcal{M}}\nu_{s,m} = 1$ and $\nu_{s,m} \geq 0$.
Since the uplink rate of the backhaul is interference-limited, the rate of the relayed signal using DF relaying is  the minimum rate of the MUE-SBS link and  SBS-MBS backhaul.
At the MBS, the received relayed signal and the direct signal are jointly combined using maximum ratio combining scheme \cite{book:tse05}
and the total throughput of the $m$-th MUE's fine message is:
\begin{equation}\label{eqn:relayedrate_OTA}
	[R_{m,F}]_{\rm OTA} = g \Bigl( \frac{1}{2} \min\Big\{[R_{ms,F}]_{\rm OTA},\nu_{s,m} [R_{s0}]_{\rm OTA}\Big\} \Bigr).
\end{equation}
The factor $\frac{1}{2}$ accounts for the half duplex DF relaying constraint and $g(\cdot)$ represents the joint combining operation at the MBS~\cite{book:tse05}. Therefore, the total MUE rate  is the sum of (\ref{eqn:rateMacroPriv}) and (\ref{eqn:relayedrate_OTA}):
\begin{equation}\label{eqn:rate_OTA}
	[R_{m}]_{\rm OTA} = [R_{m,C}]_{\rm OTA} + [R_{m,F}]_{\rm OTA}.
\end{equation}

\subsection{Delay formulation with over-the-air (OTA) backhaul}

The proposed method consists of three different links for each MUE $m$:
the direct communication between MUE and MBS, the relaying link of MUE-SBS, and the backhaul link between SBS-MBS.
Due to this, we will have different delay expressions for the above links.

We consider that the packet generation rate of MUE $m$ is divided between its coarse and fine messages as $\rho_{m,C}$ and $\rho_{m,F}$, respectively with $\rho_{m}=\rho_{m,C}+\rho_{m,F}$.
The delay for the direct MUE-MBS link is related to the coarse message of the MUE $m$ and the delay from the relayed path is based on the fine message.
Therefore, the respective delays for coarse and fine messages of MUE $m\in\mathcal{M}$ with the SBS $s\in\mathcal{S}$ as the relaying SBS, are given by:
\begin{align}\label{eqn:delayProposed}
	&[D_{m,C}]_{\rm OTA} = \frac{\rho_{m,C}}{2[R_{m,C}]_{\rm OTA}([R_{m,C}]_{\rm OTA}-\rho_{m,C})},\\
\label{eqn:delayProposed2}
	&\!\begin{multlined}[t][\displaywidth]
		[D_{m,F}]_{\rm OTA} = \biggl( \underbrace{\frac{\rho_{m,F}}{2[R_{ms,F}]_{\rm OTA}([R_{ms,F}]_{\rm OTA}-\rho_{m,F})}}_\text{MUE-SBS delay} + \\
		\underbrace{\frac{\rho_{m,F}}{2[R_{s0,F}]_{\rm OTA}([R_{s0,F}]_{\rm OTA}-\rho_{m,F})}}_\text{MBS-SBS delay} \biggr),
	\end{multlined}
\end{align}
where the rates $[R_{m,C}]_{\rm OTA}, [R_{ms,F}]_{\rm OTA}$ and $[R_{s0,F}]_{\rm OTA}$ are given in (\ref{eqn:rateMacroPriv}), (\ref{eqn:ratesmall_cellCom}) and (\ref{eqn:rateBackhaul}), respectively.
For a successful communication, the MBS needs to receive all the packets from the MUE in which the transmission delay depends on its largest component.
Thus the total delay of the $m$-MUE's transmission is,
\begin{equation}\label{eqn:delayProposedTotal}
	[D_{m}]_{\rm OTA}=\max\big([D_{m,C}]_{\rm OTA},[D_{m,F}]_{\rm OTA}\big).
\end{equation}

\subsection{Rate and delay formulation with wired (WRD) backhaul }

When using a wired backhaul model, the MUE transmission strategies remain analogous to the wireless backhaul case.
Therefore, the MUE rates of access links -- the links between MUEs and SBSs  and MUEs and MBS -- when using the wired backhaul model follows the same formulation as in the OTA scenario yielding:
\begin{equation}\label{eqn:accessrates_WRD}
\begin{split}
	[R_{m,C}]_{\rm WRD} &= [R_{m,C}]_{\rm OTA},\\
	[R_{m0,F}]_{\rm WRD} &= [R_{m0,F}]_{\rm OTA},\\
	[R_{ms,F}]_{\rm WRD} &= [R_{ms,F}]_{\rm OTA},
\end{split}
\end{equation}
where the defined rates represent, respectively, the coarse message overheard MBS, the fine message overheard by MBS and the fine message monitored by relaying SBS $s\in\mathcal{S}$.
The subscript ``WRD" is used to denote the wired backhaul expressions.

In the wired backhaul case, the total capacity is constrained by  $\bar{C}$ rather than by interference as in the wireless case.
This limitation subsequently affects the final rate of the relayed fine message.
Assuming the share of the capacity allocated to SBS $s$ is $C_s$ which satisfies (\ref{eqn:totwiredcapacity}), the total rate of $m$-th MUE's fine message is given by:
\begin{equation}\label{eqn:relayedrate_WRD}
	[R_{m,F}]_{\rm WRD} = g \Bigl( \frac{1}{2} \min\Big\{[R_{ms,F}]_{\rm WRD},\nu_{s,m} C_{s}\Big\} \Bigr).
\end{equation}
Function $g(\cdot)$ and variable $\nu_{s,m}$ are defined as in (\ref{eqn:relayedrate_OTA}).
Using (\ref{eqn:ratesmall_cellCom}) and (\ref{eqn:relayedrate_WRD}), the total MUE rate is calculated as follows;
\begin{equation}\label{eqn:rate_WRD}
	[R_{m}]_{\rm WRD} = [R_{m,C}]_{\rm WRD} + [R_{m,F}]_{\rm WRD}.
\end{equation}

The delay calculations follows (\ref{eqn:delayProposed})-(\ref{eqn:delayProposedTotal}) based on rate-splitting.
The resulting delay formulas for the wired scenario are as follows;
\begin{align}\label{eqn:delayProposedWRD}
	&[D_{m,C}]_{\rm WRD} = \frac{\rho_{m,C}}{2[R_{m,C}]_{\rm WRD}([R_{m,C}]_{\rm WRD}-\rho_{m,C})},\\
\label{eqn:delayProposed2WRD}
	&\!\begin{multlined}[t][\displaywidth]
		[D_{m,F}]_{\rm WRD} = \biggl( \underbrace{\frac{\rho_{m,F}}{2[R_{ms,F}]_{\rm WRD}([R_{ms,F}]_{\rm WRD}-\rho_{m,F})}}_\text{MUE-SBS delay} + \\
		\underbrace{\frac{\rho_{m,F}}{2\nu_{s,m}C_s(\nu_{s,m}C_s-\rho_{m,F})}}_\text{MBS-SBS delay} \biggr),
	\end{multlined}\\
\label{eqn:delayProposedTotalWRD}
	&[D_{m}]_{\rm WRD}=\max\big([D_{m,C}]_{\rm WRD},[D_{m,F}]_{\rm WRD}\big).
\end{align}

\section{Backhaul-Aware Relaying via Reinforcement Learning}\label{sec:game}

\subsection{Game Formulation}
In the coordinated approach,  the MUEs need to autonomously select their best helping SBS given the channel conditions, as well as the underlying backhaul constraints.
Although a particular SBS can be accessed by multiple MUEs, those MUEs still have to share the limited capacity of the SBS-MBS backhaul link.
This capacity limitation restricts the number of MUEs accessing a certain SBS. 
Therefore, the MUEs have to compete among each other in order to select suitable SBSs as their helping relays.
In this regard, we formulate a noncooperative game $\mathcal{G}= \Big(\mathcal{M},\{\mathcal{A}_m\}_{m \in \mathcal{M}},\{u_m\}_{m \in \mathcal{M}}\Big)$ in which $\mathcal{M}$ denotes the set of players (i.e., the MUEs), $\{\mathcal{A}_m\}_{m \in \mathcal{M}}$ are the action sets, and $\{u_m\}_{m \in \mathcal{M}}$ are the set of utility functions of the MUEs.
The action of each MUE $m$ is composed of  its transmission power $P_m\in[0,P_{m}^{max}]$, the fraction of power allocated for its fine message $\theta_m\in[0,1]$, and the helping relay SBS $s\in\mathcal{M}$
\footnote{ Each UE judiciously selects its transmit power, $\theta_m$, and SBS.
Thus, under the proposed approach an uplink power control by is carried out by each MUE in addition to selecting a suitable SBS.}.
For  notational simplicity, we use $a_m=(P_m,\theta_m,s)$ hereafter, for all $a_m\in\mathcal{A}_m$.
The available number of actions per MUE is given by the cardinality of the action set $|\mathcal{A}_m|$.
Furthermore, for notational simplicity, we use  $u_m=u_{m}{(a_m,\boldsymbol{a}_{-m})}$ to denote the utility function of each MUE $m$ when choosing an action $a_m\in\mathcal{A}_m$.
The vector $\boldsymbol{a}_{-m}$ represents the set of actions taken by all other MUEs (players) except $m$.
In this work, we consider a utility which is a function of throughput and delay.

To capture the tradeoff between delay and throughput, we use the notion of a \emph{system power} as introduced in~\cite{PW01}.
The system power is defined as the ratio of the throughput and the delay to the power of a sensitivity factor.
By adopting the system power metric as the utility, given that the MUEs play the actions $(a_m,\boldsymbol{a}_{-m})$, the utility of MUE $m$ is calculated given by:
\begin{equation}\label{eqn:utility_tradeoff}
	u_{m}{(a_m,\boldsymbol{a}_{-m})} = \frac{\big[R_m{(a_m,\boldsymbol{a}_{-m})}\big]^{(1-\alpha)}}{\big[D_m{(a_m,\boldsymbol{a}_{-m})}\big]^{\alpha}},
\end{equation}
where $\alpha\in[0,1]$ is the sensitivity parameter reflecting the throughput-delay tradeoff.
Hereinafter, for simplicity of exposition, we will drop the dependence on the actions from the notation $R_m{(a_m,\boldsymbol{a}_{-m})}$ and $D_m{(a_m,\boldsymbol{a}_{-m})}$, and we will refer to the rate and delay simply by $R_m$ and $D_m$ respectively.

Since the system power depends on rate and delay, the above utility is influenced by the over-the-air link rates as well as the backhaul condition of the relaying SBSs.
Thus, the knowledge of the backhaul is crucial to improve the performance of MUEs when selecting SBSs as relay nodes.
With this in mind, we examine the formulated problem in both perfect and imperfect information cases.

\subsection{Perfect information}\label{sec:perfect_info}

Here, we assume that the knowledge regarding the existence of SBSs, the backhaul condition of each SBS, and the actions played by all MUEs are known to each and every MUE in the network.
We refer to this case as  \emph{perfect information}.
Based on this, each MUE selects an action with a given probability (\emph{mixed strategy}).
With perfect information, the MUEs choose specific strategies known as the optimal strategies where any deviation offers negligible utility gain for all MUEs.
The above state in which no MUE desires to change its strategy is known as an \emph{equilibrium} of the game $\mathcal{G}$.
Since the set of strategies $(\mathcal{A}_m;|\mathcal{A}_m|<\infty,\; \forall m\in\mathcal{M})$ is finite and discrete, the game $\mathcal{G}$ holds at least one equilibrium in mixed strategies \cite{pap:nash51,pap:bennis12,pap:rose11}.
In this regard, we define the \emph{$\epsilon$-coarse correlated equilibrium} ($\epsilon$-CCE) as follows \cite{pap:bennis12,pap:rose11};

\begin{definition}
	({\it $\epsilon$-coarse correlated equilibrium}):
	A mixed strategy probability $\boldsymbol{\pi}_m = \big( \pi_{m,a_m^1}, \ldots, \pi_{m,a_m^{|\mathcal{A}_m|}} \big)$ is  an $\epsilon$-coarse correlated equilibrium if, $\forall m\in\mathcal{M}$ and $\forall a_m^\prime\in\mathcal{A}_m$, where:
\begin{multline*}
	\sum_{\forall \boldsymbol{a}_{-m}\in\mathcal{A}_{-m}} \Biggl( u_m{(a_m',\boldsymbol{a}_{-m})} {\pi}_{-m,\boldsymbol{a}_{-m}} \Biggr)  \\ 
	- \sum_{\forall a_m\in\mathcal{A}_m} \Biggl( u_m{(a_m,\boldsymbol{a}_{-m})}{\pi}_{m,a_m} \Biggr) \leq \epsilon,
\end{multline*}
	where ${\pi}_{-m,\boldsymbol{a}_{-m}}=\sum_{\forall a_m\in\mathcal{A}_m}\pi(a_m,\boldsymbol{a}_{-m})$ is the marginal probability distribution w.r.t. $a_m$.
The mixed-strategy probability distribution of MUE $m$, $\boldsymbol{\pi}_m = \big( \pi_{m,a_m^1}, \ldots, \pi_{m,a_m^{|\mathcal{A}_m|}} \big)$, is defined as,
\begin{align}
	&\pi_{m,a_m^l} = \mbox{Pr}(a_m=a_m^l),
\end{align}
with, $\sum_{\forall a_m^l\in\mathcal{A}_m}\pi_{m,a_m^l} = 1$, and $a_m^l$ is the $l$-th action from the set of actions $\mathcal{A}_m$ with $l\in\{1,\ldots,|\mathcal{A}_m|\}$.
\end{definition}

The motivation behind introducing the $\epsilon$-CCE is the relation between the regret matching technique -- an iterative learning mechanism -- and the existence of the $\epsilon$-CCE.
Authors in \cite{pap:hart97} claim that if all the players in a game follow the adaptive procedure of regret matching, then almost surely the game converges to the $\epsilon$-CCE. 
The regret matching mechanism allows players in a game to explore all their actions and learn the optimal strategies over time.
Moreover, the correlated equilibrium is a generalization of Nash equilibrium which allows players to coordinate their actions in order to provide better overall performance compared to the Nash approach~\cite{pap:rose11,pap:bennis12}. 
Therefore, the correlated equilibrium is more relevant in the context of decentralization and dense networks over the Nash equilibrium.

Using the regret matching technique for our considered scenario with perfect information, every MUE $m$ calculates its regret for an action $a_m^l\in\mathcal{A}_m$ at a given time instance $t$ as follows:
\begin{equation}\label{eqn:regretFull}
	r_{m,a_m^{l}}(t) = \frac{1}{t}\sum_{n=1}^t\biggl(u_m^{(n)}{(a_m^l,\boldsymbol{a}_{-m})} - u_{m}^{(n)}{(a_m,\boldsymbol{a}_{-m})}\biggr),
\end{equation}
where $(a_m,\boldsymbol{a}_{-m})$ are the actions played by the all MUEs at time $t$ and $u_m^{(n)}(\cdot)$ is the utility of MUE $m$ at time $n$.
The regret vector for MUE $m$ for actions $(a_m^1,\ldots,a_m^{|\mathcal{A}_m|})$ at time $t$ is $\boldsymbol{r}_m(t)=\big(r_{m,a_m^1}(t),\ldots,r_{m,a_m^{|\mathcal{A}_m|}}(t)\big)$.
A positive regret for a given action ($r_{m,a_m^l}(t)>0$) implies that the MUE could have obtained a higher payoff by playing this action during previous time instants, while a negative or zero regret implies that the MUE has no regret for playing that action.
Therefore, for the calculated regrets $\boldsymbol{r}_m(t)$, MUE $m$ tends to select the action with highest regret in which the mixed strategy probabilities are given as follows;
\begin{equation}\label{eqn:regretFullprob}
	\pi_{m,a_m^{l}}(t) = \frac{r_{m,a_m^l}^+(t)}{\displaystyle\sum_{\forall a_m^{l'}\in\mathcal{A}_m} r_{m,a_m^{l'}}^+(t)},
\end{equation}
where $r_{m,a_m^l}^+(t)=\max\big(0,r_{m,a_m^l}(t)\big)$ and  $\boldsymbol{\pi}_m(t)=\big(\pi_{m,a_m^{1}}(t),\ldots,\pi_{m,a_m^{|\mathcal{A}_m|}}(t)\big)$ is the mixed strategy probability profile of MUE $m$ at time $t$.

Under perfect information, the MUEs compute their utilities for all actions by observing the actions played by the other  MUEs at any time instant and achieve the CCE ($\epsilon=0$) with the above regret matching mechanism \cite{pap:rose11,pap:hart97}.
Clearly, this requires a significant amount of information which is often unavailable or complex to gather in a practical wireless network.
Therefore, it is of interest to extend this approach to handle cases in which only partial information is available at the MUEs.
Next, we propose a distributed solution which only relies on individual (and possibly noisy) feedbacks of MUE rates to optimize its utility function. 
In the sequel, we refer to this case as imperfect information.

\subsection{Imperfect information }\label{subsec:learn}

Here, the MUEs have no information about the actions of other MUEs in the network, i.e. there is no information exchange among all the MUEs, MBS and SBSs.
Due to this lack of information, the MUEs are unable to update their mixed strategies by calculating their regrets as per (\ref{eqn:regretFull}).
In order to overcome this problem, we propose a distributed learning mechanism in which the MUEs coordinate their own transmissions in an implicit manner and without information exchange.
This coordination procedure only relies on a feedback of MUEs' individual rates sent by the MBS and/or SBSs.
In particular, at each time instant, MUEs autonomously choose their own actions, receive a feedback, and build a probability distribution function over their transmission strategies.
The feedback received by every MUE is used to estimate its utility, and subsequently  compute its regret for playing a given action.
This procedure is carried out till convergence to the $\epsilon$-CCE.

For any time instant $t$ and for all $m\in\mathcal{M}$, MUE $m$ selects an action from $\mathcal{A}_m$ following the probability distribution available from the previous stage $t-1$,
\begin{equation}\label{eqn:probability}
	\boldsymbol{\pi}_m(t-1) = \big( \pi_{m,a_m^{1}}(t-1), \ldots, \pi_{m,a_m^{|\mathcal{A}_m|}}(t-1) \big).
\end{equation}
After all MUEs select and play their actions, each of them receives feedbacks of their utilities given by:
\begin{equation}\label{eqn:feedback}
	\tilde{u}_m(t) = u_{m}(t) + n_m,\; \forall m\in\mathcal{M},
\end{equation}
where $u_{m}(t)$ is the actual value of the utility as measured at the MBS and $n_m$ is the additive noise with zero mean.
In detail, each MUE $m\in\mathcal{M}$ estimates its regret $\tilde{\boldsymbol{r}}_m(t)=\big(\tilde{r}_{m,a_m^1},\ldots,\tilde{r}_{m,a_m^{|\mathcal{A}_m|}}\big)$ for all actions based on the accumulated history, given in (\ref{eqn:probability}) and (\ref{eqn:feedback}).
In addition, the MUEs must balance between the actions yielding higher regrets and exploring other actions with lower regrets with non-zero probability \cite{pap:bennis11,pap:bennis12}.
Such a behavior is captured by the so-called Boltzmann-Gibbs (BG) distribution $\boldsymbol{\beta}_{m}\big(\tilde{\boldsymbol{r}}_m(t)\big)$ which is the solution of the following optimization problem:
\begin{multline}\label{eqn:BGopt}
	\boldsymbol{\beta}_{m}\big(\tilde{\boldsymbol{r}}_m(t)\big) \in \arg \max_{\boldsymbol{\pi}_m\in\Delta(\mathcal{A}_m)} \Bigl[ \sum_{\forall a_m\in\mathcal{A}_m} \Big\{ \pi_{m,a_m}\tilde{r}_{m,a_m}^+(t) \\
	- \kappa_m\pi_{m,a_m}\ln(\pi_{m,a_m}) \Big\} \Bigr]
\end{multline}
where
$\kappa_m>0$ is a temperature parameter which balances between exploration and exploitation.

It is worth noting that allowing $\kappa_m\rightarrow 0$ for all $m\in\mathcal{M}$ leads to maximizing the sum of regrets $\sum_{\forall a_m\in\mathcal{A}_m} \bigl\{ \pi_{m,a_m}\tilde{r}_{m,a_m}^+(t) \bigl\}$.
The result is the mixed strategy in which MUEs have exploited the actions with higher regrets at the time period $t$.
This can lead to actions having zero probability and, consequently, MUE $m$ will no longer attempt to choose such actions in the future.
Conversely, maximizing the entropy ($\sum_{\forall a_m\in\mathcal{A}_m} \bigl\{ -\pi_{m,a_m}\ln(\pi_{m,a_m}) \bigl\}$) alone by allowing $\kappa_m\rightarrow\infty$, for all $m\in\mathcal{M}$, will result in an uniform distribution over the action set in which each action is equally  played.
By judiciously combining the mixed-strategy regret and the entropy regulated by the temperature parameter $\kappa_m$, we obtain a mixed strategy probability that exploits certain actions to maximize the expected utility while providing an opportunity to explore the rest of the actions.
As a result, users maximize their long-term utility metric $\overline{u}_{m}(t)=\frac{1}{t}\sum_{n=1}^{t}u_m(n)$.

For a given set of $\tilde{\boldsymbol{r}}_m(t)$ and $\kappa_m$ values, we solve the continuous and strictly concave optimization problem given in (\ref{eqn:BGopt}). The resulting probability distribution for MUE $m$ is given as follows:
\begin{equation}\label{eqn:BG}
	\beta_{m,a_m}\big(\tilde{\boldsymbol{r}}_m(t)\big) = \frac{\exp\big(\frac{1}{\kappa_m}\tilde{r}^+_{m,a_m}(t)\big)} {\sum_{\forall \acute{a}_m\in\mathcal{A}_m}\exp\big(\frac{1}{\kappa_m}\tilde{r}^+_{m,\acute{a}_m}(t)\big)} ,~\forall a_m\in\mathcal{A}_m
\end{equation}
where $\beta_{m,a_m}\big(\tilde{\boldsymbol{r}}_m(t)\big)$ is the element of $\boldsymbol{\beta}_{m}\big(\tilde{\boldsymbol{r}}_m(t)\big)$ related to the action $a_m$.

Let $\hat{u}_{m,a_m^l}(t)$ be the estimated utility of MUE $m$ at time $t$ for action $a_m^l$ and the estimated utility vector for all the set of actions is $\hat{u}_{m}(t)=\big(\hat{u}_{m,a_m^1}(t),\ldots,\hat{u}_{m,a_m^{|\mathcal{A}_m|}}(t)\big)$.
Similarly, the regret estimation vector of MUE $m$ for each action is $\tilde{\boldsymbol{r}}_{m}(t)=\big(\tilde{r}_{m,a_m^1}(t),\ldots,\tilde{r}_{m,a_m^{|\mathcal{A}_m|}}(t)\big)$ which is used to calculate the BG distribution in (\ref{eqn:BG}).
At each time instant, given that the action played by MUE $m$ at time $t$ is $a_m$ is governed by (\ref{eqn:probability}), the estimations of utility, regrets and probability distribution functions caried out by each MUE $m$ are updated for all $a_m^l\in\mathcal{A}_m$ as follows:
\begin{equation}\label{eqn:algoUpdates}
\begin{cases}
	\hat{u}_{m,a_m^l}(t) &= \hat{u}_{m,a_m^l}(t-1)  \\ 
	&\hfill + \lambda_m(t)\mathds{1}_{\{a_m^l=a_m\}} \Bigl(\tilde{u}_m(t)-\hat{u}_{m,a_m^l}(t-1)\Bigr),\\
	\tilde{r}_{m,a_m^l}(t) &= \tilde{r}_{m,a_m^l}(t-1)  \\
	&\hfill +\gamma_m(t) \Bigl(\hat{u}_{m,a_m^l}(t)-\tilde{u}_m(t)-\tilde{r}_{m,a_m^l}(t-1)\Bigr),\\
	\pi_{m,a_m^l}(t) &= \pi_{m,a_m^l}(t-1) \\
	&\hfill + \mu_m(t) \Bigl(\beta_{m,a_m^l}\big(\tilde{\boldsymbol{r}}_m(t)\big)-\pi_{m,a_m^l}(t-1)\Bigr),
\end{cases}
\end{equation}
where the learning rates $\lambda_m(t), \gamma_m(t)$ and $\mu_m(t)$ satisfy following conditions;
\small
\begin{multline}\label{eqn:learningRateConditions}
	(i) \quad \displaystyle\lim_{t\rightarrow\infty}\sum_{n=1}^t\lambda(n) = +\infty, \quad   \displaystyle\lim_{t\rightarrow\infty}\sum_{n=1}^t\gamma(n) = +\infty \\
	 \mbox{and} \quad \displaystyle\lim_{t\rightarrow\infty}\sum_{n=1}^t\mu(n) = +\infty.  
\end{multline}
\vspace{-7 mm}
\begin{multline}
	(ii) \quad \displaystyle\lim_{t\rightarrow\infty}\sum_{n=1}^t\lambda^2(n) < +\infty, \quad \displaystyle\lim_{t\rightarrow\infty}\sum_{n=1}^t\gamma^2(n) < +\infty, \\
	\mbox{and} \quad \displaystyle\lim_{t\rightarrow\infty}\sum_{n=1}^t\mu^2(n) < +\infty.
\end{multline}
\normalsize
\vspace{-7 mm}
\begin{multline}
(iii) \quad \displaystyle\lim_{t\rightarrow\infty}\frac{\gamma(t)}{\lambda(t)} = 0 \quad \mbox{and} \quad \displaystyle\lim_{t\rightarrow\infty}\frac{\mu(t)}{\gamma(t)} = 0.
\end{multline}


\begin{proposition}\label{thm:converge}
	\emph{(Convergence)}:
	The learning algorithm presented in (\ref{eqn:algoUpdates}) converges to the $\epsilon$-CCE if and only if for all $m\in\mathcal{M}$ the conditions (\ref{eqn:learningRateConditions}) are satisfied and,
 \begin{align}\label{eqn:convg}
	&\lim_{t\rightarrow\infty}\boldsymbol{\pi}_m(t) = \boldsymbol{\pi}_m^*, \\
	&\lim_{t\rightarrow\infty}\overline{u}_{m}(t) = \breve{u}_m^{(\boldsymbol{\pi}_m^*,\boldsymbol{\pi}_{-m}^*)},
\end{align}
where $\boldsymbol{\pi}^*=\big(\boldsymbol{\pi}_1^*,\ldots,\boldsymbol{\pi}_M^*\big)$ is the $\epsilon$-CCE strategy profile of the game
$\mathcal{G}= \Big(\mathcal{M},\{\mathcal{A}_m\}_{m \in \mathcal{M}},\{u_m\}_{m \in \mathcal{M}}\Big)$
 and $\breve{u}_m^{(\boldsymbol{\pi}_m^*,\boldsymbol{\pi}_{-m}^*)}$ is the optimal expected time-averaged utility of MUE $m\in\mathcal{M}$.
\end{proposition}

\begin{proof}
See Appendix \ref{sec:appndx}.
\end{proof}

\section{Simulation Results}\label{sec:results}

\begin{table}[!t]
\caption{Simulation parameters.}
\begin{center}
\begin{tabular}{l c}
\hline
{\bf Parameter} & {\bf Value} \\
\hline \hline
Carrier frequency	& $1.85$ GHz \\
System bandwidth	& $5$ MHz \\
Thermal noise ($N_0$) &$-174$ dBm/Hz \\
Wired backhaul capacity ($\overline{C}$) & 50 Mbps\\
Packet generation rate ($\rho$) & 180 kbps\\
UE transmission power & 21 dBm \\
SBS transmission power & 30 dBm \\
\hline \multicolumn{2}{c}{\bf Macro and Small cells} \\ \hline
Macrocell, small cell radius & 400 m, 50m \\
Sectors per macrocell & 3 \\
Maximum MUEs per sector & 30 \\
Maximum SBSs per sector & 15 \\
\hline \multicolumn{2}{c}{\bf Path loss models} \\ \hline
MUE - MBS & $15.3 + 37.6\log(d)$\\
MUE - SBS & $15.3 + 37.6\log(d)$\\
Shadowing & 10 dB \\
\hline \multicolumn{2}{c}{\bf Learning} \\ \hline
Boltzmann temperature ($\kappa$) & 10\\
Sensitivity between rate and delay ($\alpha$) & 0.5 \\
learning rates: $\lambda(t), \gamma(t)~\mbox{and}~\mu(t)$ & $\frac{1}{(t+1)^{.50}}, \frac{1}{(t+1)^{.55}}~\mbox{and}~\frac{1}{(t+1)^{.60}}$ \\
\hline
\end{tabular}
\end{center}
\label{tab:sim_para}
\end{table}

We consider a single macrocell with radius $R_m=400$ m which is divided into three $120^{\circ}$ sectors. Within each sector,  a number of  MUEs and SBSs with radius $R_s=50$~m  are randomly deployed.
The maximum uplink transmission powers for MUEs and SBSs are set to $21$~dBm and $30$~dBm, respectively~\cite{onln:3gpp06}.
We use the 3GPP specifications for pathloss and shadowing in outdoor links with the noise level set to -$174$~dBm~\cite{onln:3gpp06}.
A detailed list of the simulation parameters are given in Table \ref{tab:sim_para}.
All statistical results are averaged over the random channel variations and locations of the users and base stations.

Throughout the simulations, we study the performance of the proposed learning mechanism, and compare the results with four baseline algorithms as described in Table \ref{tab:models}.

\begin{table}[!t]
\centering
\renewcommand\multirowsetup{\centering}
\caption{Proposed and Benchmark Algorithms}
\label{tab:models}
\begin{tabular}{p{1 cm} p{2.25 cm} p{4.35 cm}}
\cmidrule[0.75 pt]{2-3}
& \makebox[2.25 cm][c]{{\bf Model}} & \makebox[4.25 cm][c]{{\bf Description}} \\
\cmidrule[0.75 pt]{2-3}
\ldelim \{ {12}{1 cm}[\parbox{1 cm}{Proposed Models}]
&
Proposed learning approach with full information (RS-F) & The proposed mechanism with the assumption of full information availability. The solution is governed using (\ref{eqn:regretFull}) and (\ref{eqn:regretFullprob}). \\
& 
Proposed learning approach with imperfect information (RS-L) & The proposed learning approach with the imperfect information availability where the MUEs have no knowledge about the rest of MUEs, their actions and the capacity of the backhaul network. The MUEs learn their best strategies only with the aid of feedbacks of their own data rates. \\
\cline{2-3} 
\ldelim \{ {25}{1 cm}[\parbox{1 cm}{Baseline Models}]
&
Classical approach (CLA) & This is the classical two-tier network scenario, in which there is no coordination among macro- and small cell tiers. The objective function of every MUE is given in (\ref{eqn:noncop_mbsrate}) and (\ref{eqn:noncop_SBSrate}). \\
& 
Reuse-1 approach (RU-1) & Similar to CLA, the reuse-1 approach considers no cooperation between two-tiers. Here, MUEs uniformly split their transmit power over the entire bandwidth. \\
& Offloading with full information (OF-F) & Here, some MUEs are fully handed over to the SBSs in order to maximize their individual data rates. The rationale of  OF-F is to compare the performance of the rate-splitting technique with proposed learning method to the classical offloading. \\
& Satisfaction based learning (SAT) & This learning algorithm is proposed in \cite{pap:meyem13}. Here, every MUE is not interested in maximizing its utility function, but rather being satisfied based on achieving a predefined satisfaction level. This is used to compare the learning capability of the proposed RS-L method.\\
\cmidrule[0.75 pt]{2-3}
\end{tabular}
\end{table}


\subsection{Performance comparison}

\begin{figure}[!t]
\centering
\includegraphics[width=.5\textwidth]{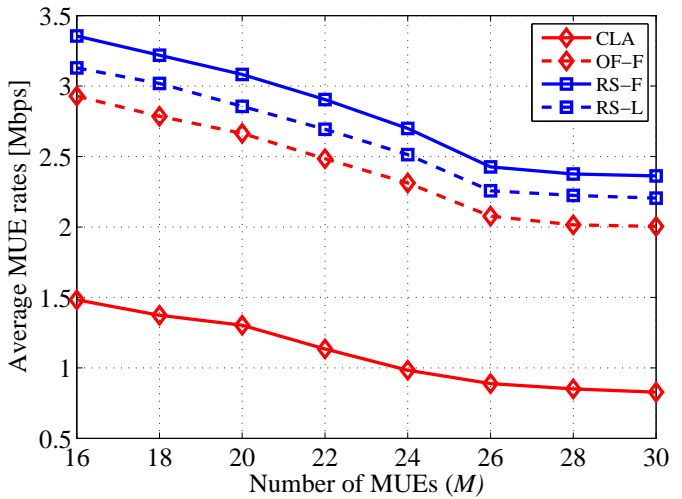}
\caption{Average MUE rate as a function of the number of MUEs ($S=8$, Wireless backhaul).}
\label{fig:rateMUEchng}
\end{figure}

Fig.~\ref{fig:rateMUEchng} shows the average rate per MUE as the number of MUEs increases.
The results are obtained for the proposed approach with perfect and imperfect information availability (RS-F and RS-L) and compared to the CLA and OF-F methods with an over-the-air backhaul.
The offloading method with perfect information (OF-F) achieves almost twice the rates compared to the classical implementation, while the proposed method yields an additional improvement of $10\%$.
Fig.~\ref{fig:rateMUEchng} shows that the proposed RS-L approach with limited information achieves comparable rates with the perfect information case.
As the number of MUEs increases, a drop in the average MUE rates can be seen for all four cases.
The reason is that the shared spectrum with fixed amount of sub-carriers suffers from the additional interference due to the increased number of MUEs.
More interference decreases the signal-to-interference-and-noise-ratio resulting in lower throughput for all MUEs.

\begin{figure}[!t]
\centering
\includegraphics[width=.5\textwidth]{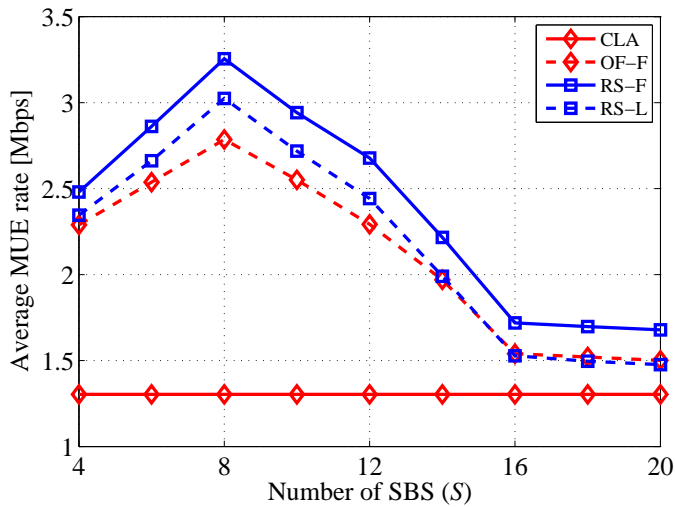}
\caption{Average MUE rate as a function of the number of SBSs ($M=20$, Wireless backhaul).}
\label{fig:rateSBSchng}
\end{figure}

Fig.~\ref{fig:rateSBSchng} shows the average MUE rates as the number of SBSs varies for a network with $N=16$ sub-carriers where $8$ sub-carriers allocated for MUEs and the rest are allocated for the backhaul.
In the absence of cross-tier cooperation, changing the number of SBSs has no impact on the performance.
In contrast in Fig.~\ref{fig:rateSBSchng} we see that the network size affects the proposed and offloading approaches.
As the number of SBSs increases, the MUEs have more opportunities to use certain SBSs as relays, increasing thereby MUEs transmission rates.
Fig.~\ref{fig:rateSBSchng} shows that the OF-F and RS-F schemes yield a significant rate improvement compared to the CLA reaching up to 2.5 times improvement with $S=8$.
However, as the number of SBSs increases, the MUE rates start to decrease.
This decrease is mainly due to the interference limitations over the wireless backhaul.
In Fig. 3, we see that the proposed approach with perfect information, RS-F, yields a notable advantage compared to OF-F, reaching up to $10\%$ of improvement at $S>8$, in terms of MUE rates.
This result demonstrates that the proposed approach is able to better handle the limitations of the backhaul.
In addition, Fig.~\ref{fig:rateSBSchng} shows that the proposed distributed method with imperfect information, RS-L, achieves around $96\%$ improvement in terms of MUE rates, compared to its upper bound which is RS-F.


\subsection{Backhaul Impact}

In order to compare the backhaul models, we plot the cumulative density functions (CDFs) of rates and delays.
We consider two reference models: CLA and RU-1, along with the proposed approach.
The two backhaul models are described in Section \ref{sec:network_model} and the rate-splitting with wireless backhaul is referred to as ``RS-OTA" while ``RS-WRD'' represents the wired backhaul.
Furthermore, we assume a hybrid backhaul in which both the wired and wireless backhauls are dynamically selected.
The simulation carried out for the system with rate-splitting technique and a hybrid backhaul is referred to as ``RS-HYB" hereafter.
In order to evaluate the proposed learning algorithm, the satisfaction based learning approach proposed in \cite{pap:meyem13} with a hybrid backhaul (SAT-HYB) is used as the third baseline model.

\begin{figure}[!t]
\centering
\subfigure[Cumulative density function of MUEs rates.]{
	\includegraphics[width=0.5\textwidth]{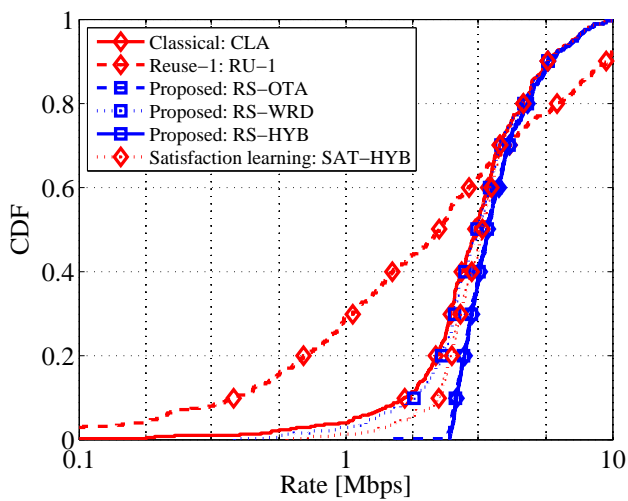}
	\label{fig:rate_cdf}
}
\hfil
\subfigure[Cumulative density function of MUE delays.]{
	\includegraphics[width=0.475\textwidth]{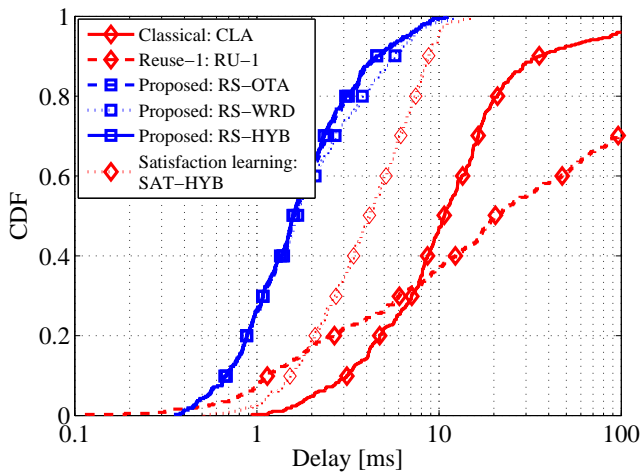}
	\label{fig:delay_cdf}
}
\caption{Cumulative density functions of MUE rates and delays comparing six schemes classical (CLA), reuse-1 (RU-1), proposed rate splitting with wireless (OTA), wired (WRD) and hybrid (HYB) backhaul networks, and satisfaction based learning with hybrid backhaul (SAT-HYB) for the 10\% best effort MUEs.}
\label{fig:cdfs}
\end{figure}

Fig.~\ref{fig:rate_cdf} and Fig.~\ref{fig:delay_cdf} show the rate and delay CDFs achieved for a network with $M=15$ and $S=8$.
We consider the best effort MUEs -- the best 10\% -- for both plots in Fig.~\ref{fig:cdfs}.
Using the RU-1 method, we can see that the best effort MUEs exhibit large rates and delay variations.
In contrast, using the CLA and proposed techniques, most MUEs are able to achieve a minimum rate level and low delays.
For both CLA and RS-WRD, the best effort MUEs have almost equivalent rates and the respective average rates are $3.61$ and $3.66$ Mbps.
The best effort users using RS-OTA achieve higher rates with a mean of $3.96$ Mbps while the average rate achieved with RS-HYB is $3.98$ Mbps.
Clearly, Fig.~\ref{fig:rate_cdf} shows that, in terms of rates, RS-HYB outperforms all the other techniques.
However, in terms of delays, the proposed approach with any type of backhaul achieves a similar distribution for best effort users.
The corresponding average delays for RS-OTA, RS-WRD and RS-HYB are $1.6, 1.8$ and $1.6$ ms.
In this respect, we can see that these proposed approaches exhibit a significant reduction in delay when compared to the CLA case which achieves an average delay of $34$ ms.
This is due to the fact that the proposed approach allows the MUEs to achieve higher rates and lower delays by leveraging nearby SBSs as relays and suitable backhaul condition.
The two curves -- RS-HYB and SAT-HYB -- allow us to compare the proposed learning algorithm to the satisfaction based learning mechanism.
Both rate and delay curves show that the proposed solution allows the MUEs to achieve higher rates with lower delays.
This is due to the fact that the MUEs in RS-HYB always attempt to improve their utility while MUEs from SAT-HYB aim to achieve a fixed, pre-determined satisfaction level with the system power of 5 for all $m\in\mathcal{M}$.
Furthermore, we note that Fig.~\ref{fig:cdfs} shows that by smartly exploiting the hybrid backhaul via the proposed RS-HYB method, significant performance gains can be achieved as opposed to the benchmark methods.

\begin{figure}[ht]
\centering
\includegraphics[width=.5\textwidth]{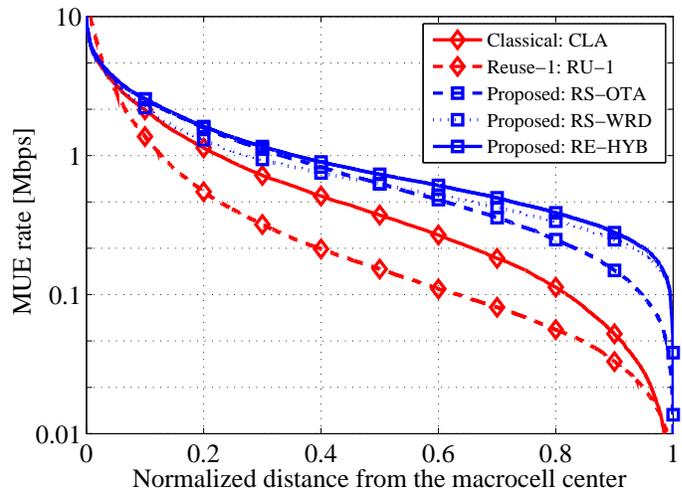}
\caption{Average MUE throughput as a function of the distance from MBS to MUE.}
\label{fig:coverage}
\end{figure}

In Fig. \ref{fig:coverage}, we further assess the performance of the proposed approach under different locations for the MUEs.
X-axis represents the normalized distance from the MBS (or cell center) to any MUE while the y-axis displays the average throughput achieved by the MUE at a given distance.
An MUE close to the MBS, can achieve higher rates for all five approaches due to the high channel gain over the MUE-MBS link.
However, as the MUE moves away from the MBS, the channel gain of MUE-MBS link degrades proportional to the distance based on the path loss model resulting in decreased rates.
Moreover, due to its ability to assign MUEs over different sub-carriers, the CLA schemes achieves a higher throughput than the RU-1 approach which allows MUEs to transmit over all sub-carriers. 
As the MUE moves toward the cell edge, the proposed learning methods -- RS-OTA, RS-WRD and RS-HYB -- yield a significant performance advantage as they allow to leverage the use of SBSs as relays.
In particular, as the MUE-SBS link quality improves, the MUEs are able to achieve a higher rate for their fine message, thus yielding an increased performance advantage for the proposed schemes.

In addition, Fig.~\ref{fig:coverage} shows the tradeoff between wired and wireless backhauls.
Since the wired capacity is fixed and independent from the distance between the SBS and the MBS, 
the backhaul link quality remains constant.
Therefore, in the presence of a wired backhaul, the MUEs obtain an almost fixed rate for the fine message. 
However, the wireless backhaul offers a higher capacity for the SBSs close to the MBS compared to the SBSs, which are deployed far from the MBS. 
Thus, although the cell-edge MUEs leverage neighboring SBSs, the rate improvement is limited due to the low capacity backhaul link between SBS-MBS link. 
In consequence, cell-edge MUEs achieve a lower rate using RS-OTA as opposed to the wired case, i.e., using RS-WRD. 
In contrast, SBSs that are close to the MBS experience a better backhaul under a  wireless backhaul.
Therefore, compared to RS-WRD, the RS-OTA offers higher rates for MUEs close to the cell center.
The hybrid backhaul selects the best strategy over wired and wireless links.
Thus, as seen in Fig.~\ref{fig:coverage}, RS-HYB yields higher rates over the entire cell due to its ability to smartly exploit both the wired and the over-the-air backhauls.

\begin{figure}[!t]
\centering
\includegraphics[width=.5\textwidth]{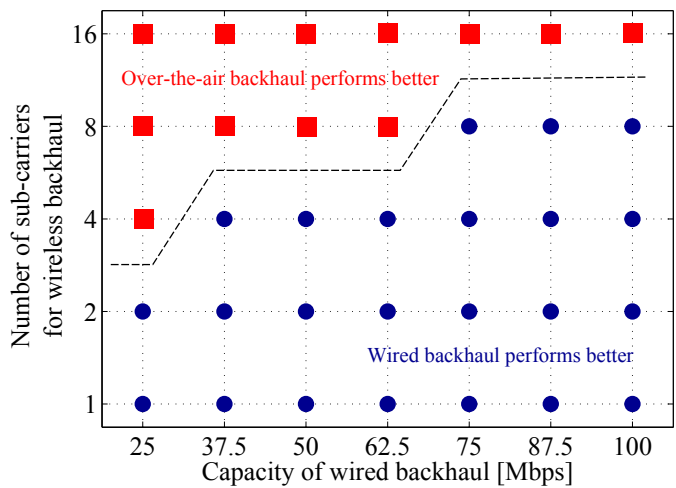}
\caption{Comparison between over-the-air backhaul with wired backhaul based on the average achievable throughputs. The upper section of the dotted line displays the over-the-air backhaul configurations, which exhibit better utilities compared to wired backhaul while the lower section illustrates vise versa.}
\label{fig:OTAvsWRD}
\end{figure}

In Fig.~\ref{fig:OTAvsWRD}, we show a comparison between over-the-air and wired backhauls depending on the bandwidth allocation for the wireless backhaul and the capacity allocation for the wired backhaul. 
For a fixed setup with M = 15 and S = 8, we compute the average MUE data rates by varying the capacity of the wired backhaul while also varying the number of sub-carriers allocated for over-the-air backhaul. 
In Fig.~\ref{fig:OTAvsWRD}, a square represents the case in which the over-the-air backhaul leads to higher MUE throughput or by a circle when the wired configuration provides highest MUE throughput.
Finally, in this figure, we observe that a wireless backhaul is preferred over the wired backhaul when $16$ or more sub-carriers are available. However, when the number of sub-carriers is below $16$, the wired backhaul yields a higher throughput.


\subsection{Convergence}

\begin{figure}[!t]
\centering
\includegraphics[width=.5\textwidth]{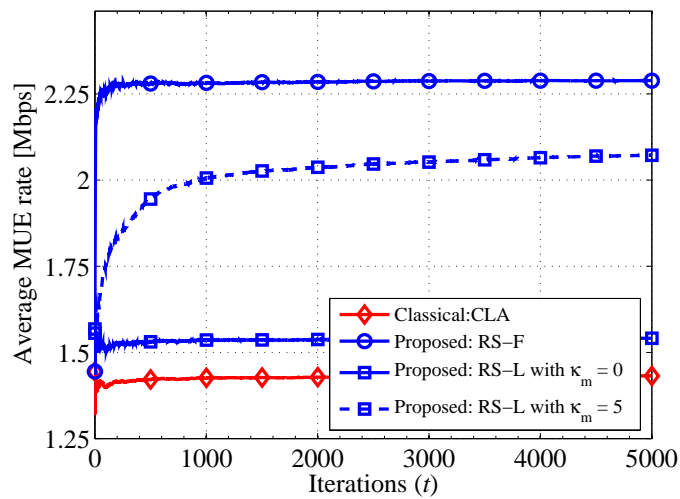}
\caption{Convergence of the proposed learning algorithm with the utility as MUE throughput ($M=10,S=8$).}
\label{fig:convg}
\end{figure}

\begin{figure}[!t]
\centering
\includegraphics[width=.5\textwidth]{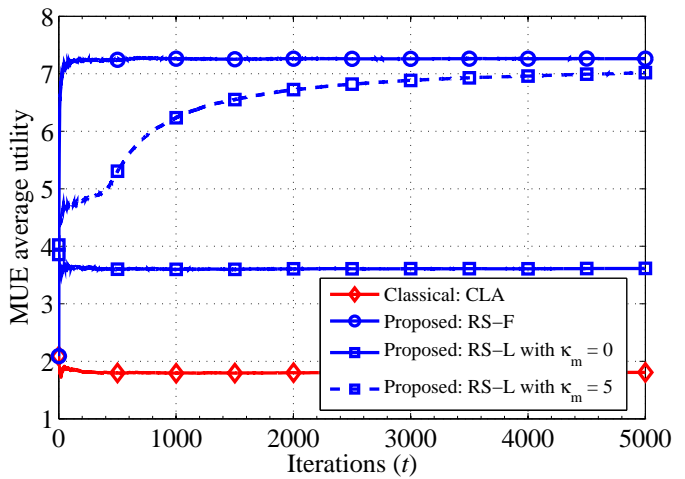}
\caption{Convergence of the proposed learning algorithm capturing the trade-off between rate and delay ($M=10,S=8$).}
\label{fig:convg_tradeoff}
\end{figure}

Next, we assess the convergence of the proposed learning scheme with incomplete information, i.e., RS-L.
Here we consider two scenarios; {\it (i)} each individual MUE attempts to maximize its throughput only and {\it (ii)} each individual MUE aims to improve its throughput while minimizing the transmission delay with the aid of the system power utility.
In Fig.~\ref{fig:convg}, we show the throughput achieved by the MUEs when $\alpha=0$ in (\ref{eqn:utility_tradeoff}) while in Fig.~\ref{fig:convg_tradeoff}, we plot the rate -- delay tradeoff  when $\alpha\in(0,1)$.
Here, $M=10$ MUEs and $S=8$ SBSs per sector and a wireless backhaul network are assumed.
From Fig.~\ref{fig:convg}, we see that the proposed RS-L algorithm converges in about $2500$ iterations.

Figures~\ref{fig:convg} and~\ref{fig:convg_tradeoff} clearly show the tradeoff between convergence time and available information.
For instance, we see that learning with incomplete information requires a higher convergence time when compared to the case with full information.
This is mainly due to: a) the noisy feedback, b) imperfect information, and c) the value of the temperature parameter $\kappa_m$ which governs the convergence speed.
A large $\kappa_m$ results in a uniform distribution for mixed strategies ($\boldsymbol{\pi}_m$) at the beginning and converges to the best mixed strategies with time.
This leads to  a smaller gap between the achievable utility between the case with full information and that with imperfect (partial) information.
Conversely, a larger $\kappa_m$ leads to MUEs exploiting their actions which are played at the beginning of the learning process, thus resulting in a faster convergence.
However, due to the lack of exploration, the resulting outcome becomes inefficient, i.e. yielding lower utilities.
In order to corroborate this,  we plot the learning curves in Fig.~\ref{fig:convg} and~\ref{fig:convg_tradeoff} with $\kappa_m=\lbrace 0,5\rbrace$. When $\kappa_m=0$, MUEs exploit their environment by playing their actions at the beginning (no exploration).
Therefore, myopic MUEs obtain quicker convergence, albeit a low utility, and for a larger $\kappa_m$ MUEs get more opportunities to explore their optimal actions at the cost of slower convergence.

\section{Conclusions}\label{sec:conclu}

In this paper, we have proposed a novel distributed reinforcement learning mechanism that allows macro-cellular users to optimize their performance by leveraging neighboring small cell base stations.
In this proposed scheme, the small cell base stations act as relays for the macro-cellular users and offload the traffic of the MUEs over a capacity limited and heterogeneous (wired, wireless) backhaul network. 
The proposed scheme allows the macro-cellular users to jointly optimize both access and backhaul links. 
We have formulated the problem as a noncooperative game and proposed a novel, fully distributed, learning algorithms that is shown to converge to a suitable equilibrium solution.
Simulation results show that our proposed algorithm allows the MUEs to improve their data rates as well as to significantly reduce their transmission delays compared to some existing benchmark approaches.
In our future work, we will study the backhaul aware downlink communication problem and the impact of carrier aggregation and multiple antennas on backhaul-aware resource management.

\appendices 
\section{Proof of Proposition \ref{thm:converge}} \label{sec:appndx}

Here, we provide the proof for $\alpha=0$ where the utility is equivalent to the MUE throughput.
Thus, we need to prove the functions $\tilde{u}_m(\cdot)$ and $\beta_{m,a_m}(\cdot)$ are Lipschitz.

For a given MUE $m\in\mathcal{M}$ at the time instance $t$, the utility $\tilde{u}_m(t)$ can be expressed in a simplified form as follows;
\begin{equation}
	\tilde{u}_m(t) = \log_2(1+x(t))
\end{equation}
where $x(t)$ is the signal-to-interference-plus-noise (SINR) of MUE $m$ at time $t$.
Without loss of generality, let us assume $x(t)>y(t)>0$ where $x(t)$ and $y(t)$ are any two possible SINR values at a given time instance.
Define $F_u = |\frac{\log_2(1+x(t))-\log_2(1+y(t))}{x(t)-y(t)}|$ which is reordered as follows; 
\begin{align}
	\nonumber F_u &= \biggl|\frac{\log_2(1+x(t))-\log_2(1+y(t))}{x(t)-y(t)}\biggr| \\
	&= \biggl|\frac{1}{(1+y(t))}\frac{\ln(1+z)}{z\ln(2)}\biggr| \leq \frac{\ln(1+z)}{z\ln(2)}
\end{align}
where $z=\frac{x(t)-y(t)}{1+y(t)}>0$.
Consider the first derivative of $\ln(1+z)/z$,
\begin{equation}
	\frac{d}{dz} \biggl(\frac{\ln(1+z)}{z}\biggr) = \frac{1}{z^2}\biggl(1-\frac{1}{1+z}-\ln(1+z)\biggr) < 0;~\forall z \geq 0,
\end{equation}
which yields the upper bound of $\frac{\ln(1+z)}{z}$ is $\lim_{z\rightarrow 0}\frac{\ln(1+z)}{z}=1$.
This concludes $F_u\leq\log_2e$ and $|\log_2(1+x(t))-\log_2(1+y(t))|\leq L|x(t)-y(t)|$ for a scaler $L$, i.e. $\tilde{u}_m(\cdot)$ is Lipschitz \cite{tech:heinonen05}.

Assume now $\boldsymbol{x}$ and $\boldsymbol{y}$ are two regret vectors of MUE $m$.
Since $\beta_{m,a_m}(\cdot)$ generates probability distributions for all regret vectors, it is obvious that following bound is true for any $\boldsymbol{x}$ and $\boldsymbol{y}$;
\begin{equation}
	0 < |\beta_{m,a_m}(\boldsymbol{x})-\beta_{m,a_m}(\boldsymbol{y})| < 1, 
\end{equation}
and it is possible to find a scaler $L$ such that $L|\boldsymbol{x}-\boldsymbol{y}|\geq 1$.
Therefore, for all $\boldsymbol{x}$, $\boldsymbol{y}$ the function $\beta_{m,a_m}(\cdot)$ satisfies $|\beta_{m,a_m}(\boldsymbol{x})-\beta_{m,a_m}(\boldsymbol{y})| \leq L|\boldsymbol{x}-\boldsymbol{y}|$ and henceforth, it is also a Lipschitz function \cite{tech:heinonen05}.

With the satisfactory of above conditions and using \cite[Equation (7) and Proposition 4.1]{notes:benaim99}, the convergence is achieved.

\vspace*{-0.2em}
\def\baselinestretch{1.2}
\bibliographystyle{IEEEtran}
\bibliography{IEEEabrv}


\vspace{-5 mm}

\begin{IEEEbiography}[{\includegraphics[width=1in,height=1.25in,clip,keepaspectratio]{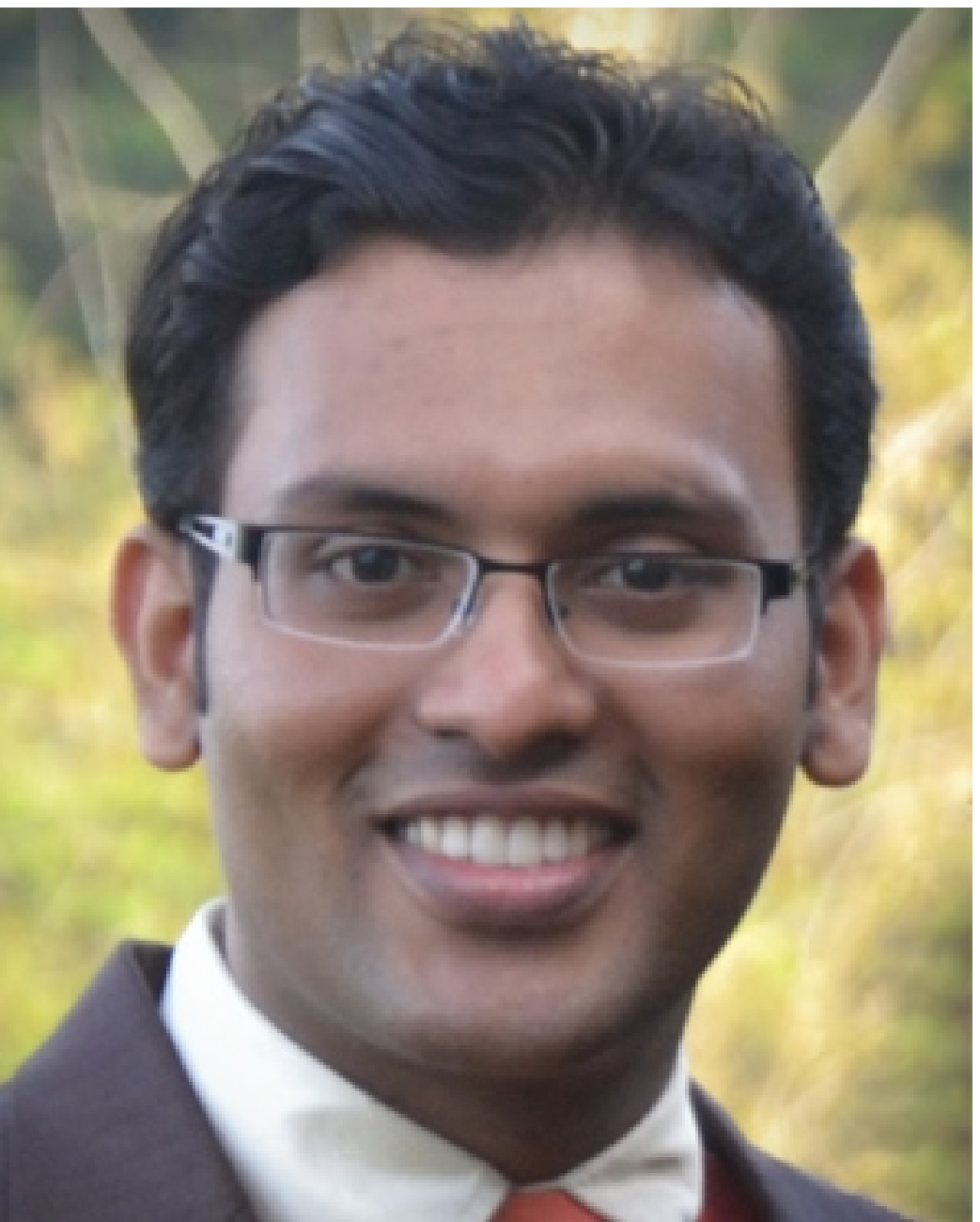}}]{Sumudu Samarakoon}
received his B. Sc. degree in Electronic and Telecommunication Engineering from the University of Moratuwa , Sri Lanka in 2009 and the M. Eng. degree from the Asian Institute of Technology, Thailand in 2011. 
He is currently working Dr. Tech (Hons.) degree in Communications Engineering in University of Oulu, Finland. 
Sumudu is also a member of the research staff of the Centre for Wireless Communications (CWC), Oulu, Finlad. 
His main research interests are in heterogeneous networks, radio resource management and game theory.
\end{IEEEbiography}

\begin{IEEEbiography}[{\includegraphics[width=1in,height=1.25in,clip,keepaspectratio]{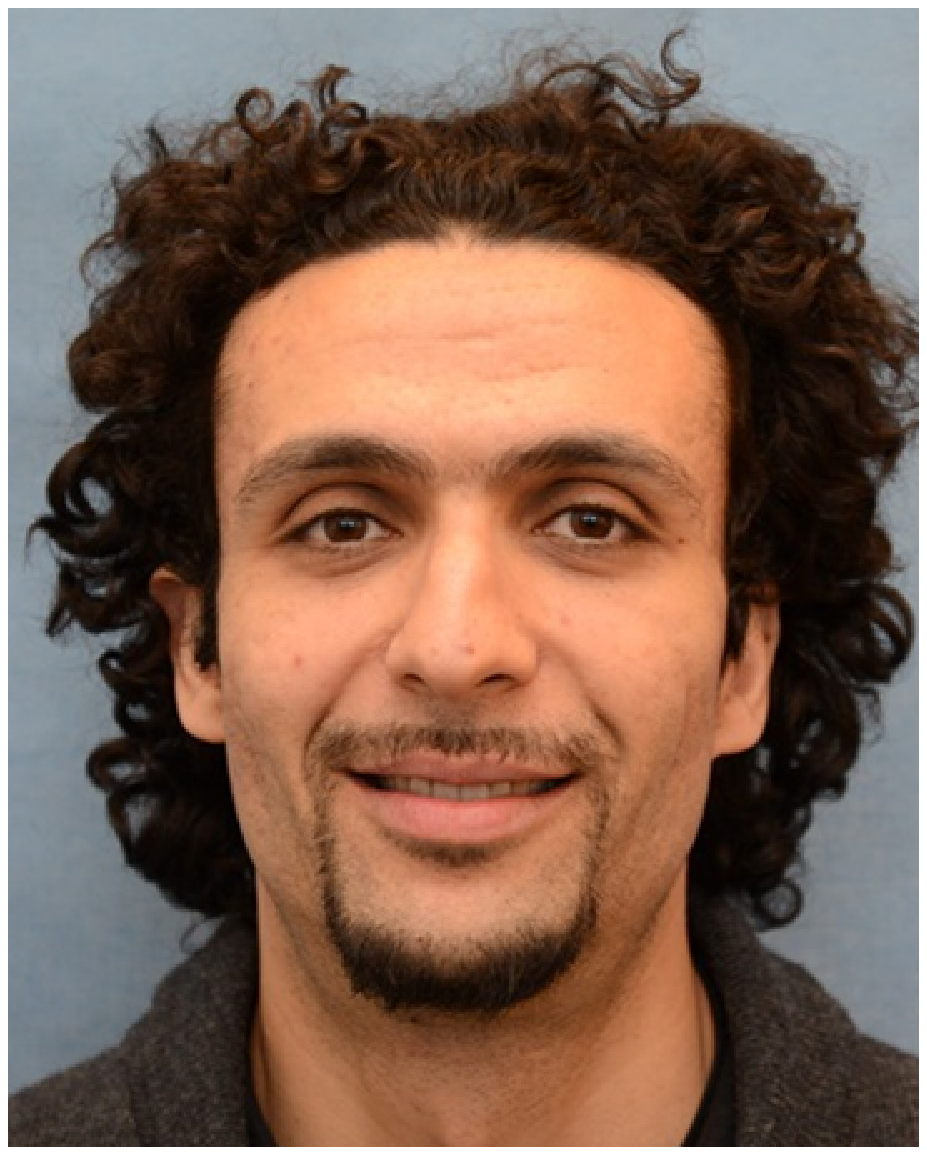}}]{Mehdi Bennis}
received his M.Sc. degree in Electrical Engineering jointly from the Ecole Polytechnique Federale de Lausanne (EPFL), Switzerland and the Eurecom Institute, France in 2002. From 2002 to 2004, he worked as a research engineer at IMRA-EUROPE investigating adaptive equalization algorithms for mobile digital TV. In 2004, he joined the Centre for Wireless Communications (CWC) at the University of Oulu, Finland as a research scientist. In 2008, he was a visiting researcher at the Alcatel-Lucent chair on flexible radio, SUPELEC. He obtained his Ph.D in December 2009 on spectrum sharing for future mobile cellular systems. He was the co-PI of the Broadband Evolved FEMTO (FP7-BeFEMTO) project, and currently the PI of the upcoming European project (CELTIC-SHARING). 

His main research interests are in radio resource management, heterogeneous networks, game theory and machine learning in the context of heterogeneous and small cell networks. Mehdi has published more than 50 research papers in international conferences, journals and book chapters. He was also a co-chair at the 1st international workshop on small cell wireless networks (SmallNets) in conjunction with ICC 2012, the 2nd Workshop on Cooperative Heterogeneous Networks (coHetNet) in conjunction with ICCCN 2012, and the upcoming 2nd international workshop on small cell wireless networks (SmallNets) in conjunction with ICC 2013 (Budapest, Hungary). Recently, he gave tutorial presentations at IEEE PIMRC 2012 (Sydney, Sep. Australia) and IEEE GLOBECOM 2012 (Annaheim, CA, Dec. 2012).
\end{IEEEbiography}
\pagebreak

\begin{IEEEbiography}[{\includegraphics[width=1in,height=1.25in,clip,keepaspectratio]{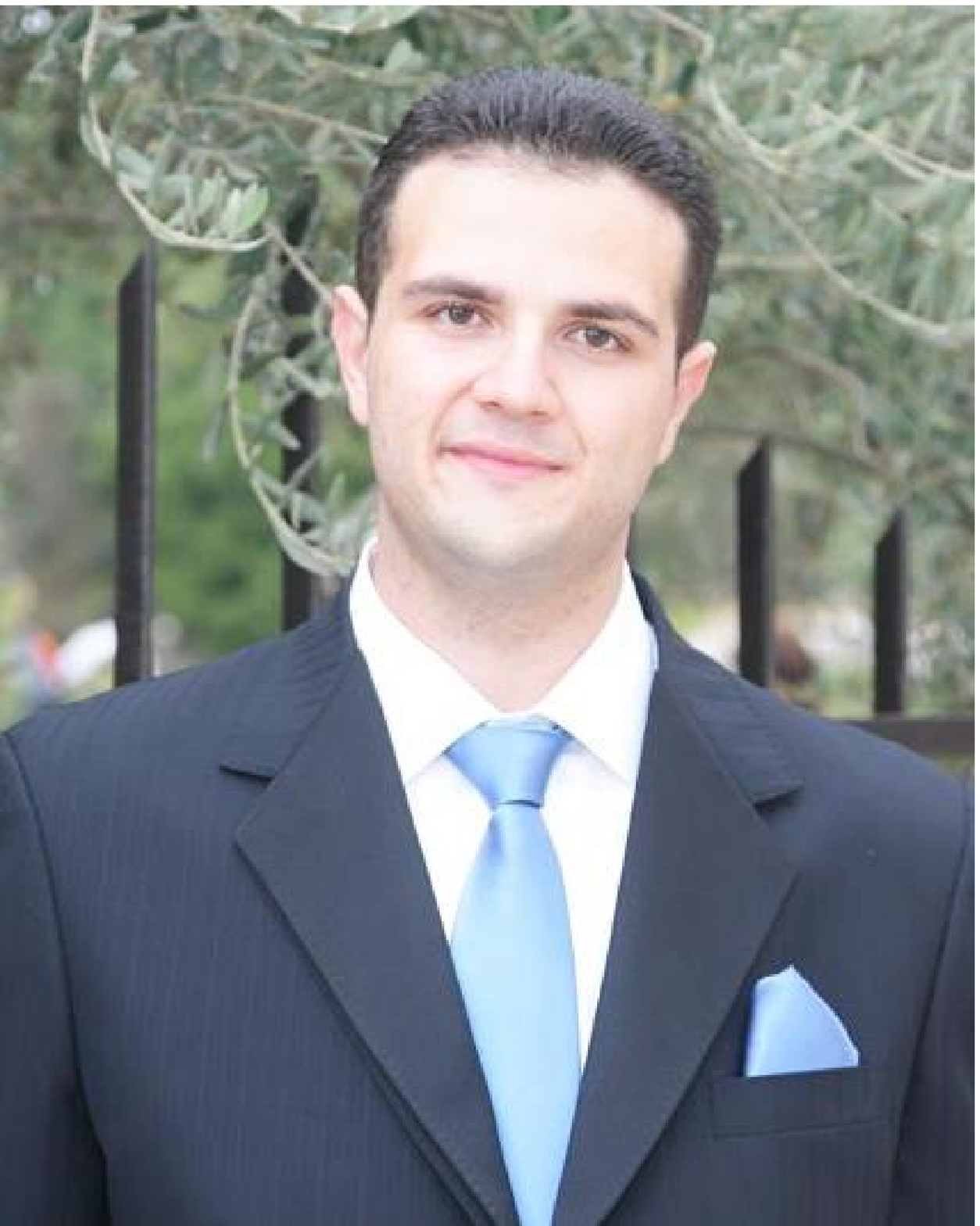}}]{Walid Saad} 
(S'07, M'10) received his B.E. degree in Computer and Communications Engineering from the Lebanese University in 2004, his M.E. in Computer and Communications Engineering from the American University of Beirut (AUB) in 2007, and his Ph.D degree from the University of Oslo in 2010. 
Currently, he is an Assistant Professor at the Electrical and Computer Engineering Department at the University of Miami. 
Prior to joining UM, he has held several research positions at institutions such as Princeton University and the University of Illinois at Urbana-Champaign. 

His research interests include wireless and small cell networks, game theory, network science, cognitive radio, wireless security, and smart grids. 
He has co-authored one book and over 75 international conference and journal publications in these areas. 
He was the author/co-author of the papers that received the Best Paper Award at the 7th International Symposium on Modeling and Optimization in Mobile, Ad Hoc and Wireless Networks (WiOpt), in June 2009, at the 5th International Conference on Internet Monitoring and Protection (ICIMP) in May 2010, and at IEEE WCNC in 2012. 
Dr. Saad is a recipient of the NSF CAREER Award in 2013.
\end{IEEEbiography}

\begin{IEEEbiography}[{\includegraphics[width=1in,height=1.25in,clip,keepaspectratio]{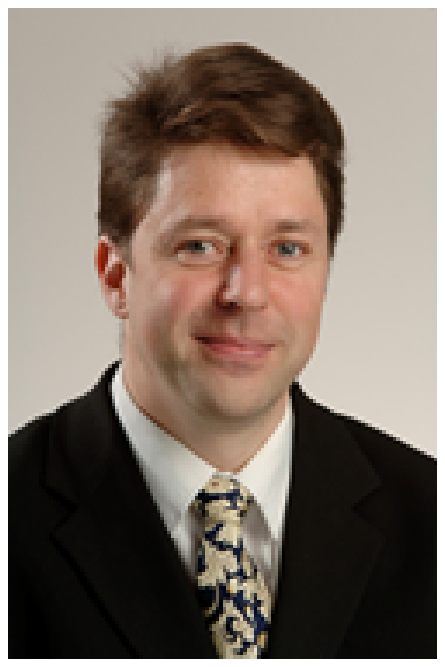}}]{Matti Latva-aho}
was born in Kuivaniemi, Finland in 1968.
He received the M.Sc., Lic.Tech. and Dr. Tech (Hons.) degrees in Electrical Engineering from the University of Oulu, Finland in 1992, 1996 and 1998, respectively.
From 1992 to 1993, he was a Research Engineer at Nokia Mobile Phones, Oulu, Finland.
During the years 1994 -- 1998 he was a Research Scientist at Telecommunication Laboratory and Centre for Wireless Communications at the University of Oulu.
Currently he is the Department Chair Professor of Digital Transmission Techniques and Head of Department at the University of Oulu, Department for Communications Engineering.
Prof. Latva-aho was Director of Centre for Wireless Communications at the University of Oulu during the years 1998-2006.
His research interests are related to mobile broadband wireless communication systems.
Prof. Latva-aho has published over 200 conference or journal papers in the field of wireless communications.
He has been TPC Chairman for PIMRC'06, TPC Co-Chairman for ChinaCom'07 and General Chairman for WPMC'08.
He acted as the Chairman and vice-chairman of IEEE Communications Finland Chapter in 2000 -- 2003. 
\end{IEEEbiography}
\vfill

\end{document}